\documentclass[11pt,a4paper]{article}
\usepackage{graphicx}
\usepackage{amsmath, amsthm, amssymb,latexsym}
\usepackage{amssymb}
\usepackage{verbatim}
\usepackage{ifthen}
\usepackage{fullpage}
\usepackage[boxed,noend]{algorithm2e}

\newcommand{\qd}{\unitlength1pt\begin{picture}(7,7)   \put(0,0){\framebox(6,6){}}\end{picture}}

\newtheorem{theorem}{Theorem}[section]

\newtheorem{lemma}[theorem]{Lemma}

\newtheorem{hypothesis}{Hypothesis}[section]

\DeclareMathOperator{\inc}{inc}
\renewcommand{\P}{\mathbf{P}}

\title{The Feedback Arc Set Problem with Triangle Inequality is a Vertex Cover Problem}
\author{Monaldo Mastrolilli\\
IDSIA, 6928 Manno, Switzerland\\
monaldo@idsia.ch}

\begin{document}
\def\ratio{r(2q-1)}
\def\rratio{r(3q-1)}
\def\GP{G_\mathbf{P}}
\def\prob{1|prec\,|\sum w_jC_j}
\def\HP{\mathbf{H}_\mathbf{P}}
\newcommand{\ssc}[0]{\ensuremath{1|\mathit{prec}|\sum w_jC_j}}
\renewcommand{\P}{\mathbf{P}}
\newcommand{\VC}[1]
{
 \ifthenelse{\equal{#1}{}}{[VC-$\HP$]}{[VC-$\HP(#1)$]}
}
\newcommand{\MinFAS}{\textsc{MinFAS}}
\date{}

\maketitle
\begin{abstract}
 We consider the (precedence constrained) Minimum Feedback Arc Set problem with triangle inequalities on the weights, which finds important applications in problems of ranking with inconsistent information. We present a surprising structural insight showing that the problem is a special case of the minimum vertex cover in hypergraphs with edges of size at most 3. This result leads to combinatorial approximation algorithms for the problem and opens the road to studying the problem as a vertex cover problem.
\end{abstract}
%%%%%%%%%%%%%%%%%%%%%%%%%%%%%%%%%%%%%%%%%%%%%%%%
%%%%%%%%%%%%%%%%%%%%%%%%%%%%%%%%%%%%%%%%%%%%%%%%
\section{Introduction}

The \textsc{Minimum Feedback Arc Set} problem (\textsc{MinFAS}) is a fundamental and classical combinatorial optimization problem that finds application
in many different settings that span from circuit design, constraint satisfaction problems, artificial intelligence, scheduling, etc. (see e.g. Chapter 4 in~\cite{PardalosDu99:Handbook} for a survey). For this reason it has been deeply studied since the late 60's (see, e.g., \cite{lempel66}).

Its input consists of a set of vertices
$V$ and nonnegative weights $\{w_{(i,j)}: (i,j) \in V\times V\}$ for every oriented pair of vertices. The goal is to find a permutation
$\pi$ that minimizes $\sum_{\pi(i)<\pi(j)} w_{(i,j)}$, i.e. the weight of pairs of vertices that comply with the permutation\footnote{Different, but equivalent formulations are often given for the problem. Usually the goal is defined as the minimization of the
weight of pairs of vertices \emph{out of order} with respect to the permutation, i.e. $\sum_{\pi(i)<\pi(j)} w_{(j,i)}$. Clearly by swapping appropriately the weights we obtain the equivalence of the two definitions.}.
A \emph{partially ordered set} (\emph{poset}) $\mathbf{P} = (V,P)$,
consists of a set $V$ and a partial order $P$ on $V$, i.e., a
reflexive, antisymmetric, and transitive binary relation $P$ on $V$, which indicates that, for certain pairs of elements in the set, one of the elements precedes the other.
In the \emph{constrained} \textsc{MinFAS} (see~\cite{ZuylenHJW07}) we are given a partially ordered set $\mathbf{P} = (V,P)$  and we want to find a linear extension of $\mathbf{P}$ of minimum weight.

%%%%%%%%%%%%% LITERATURE %%%%%%%%%%%%
%%%%%%%%%%%%% LITERATURE %%%%%%%%%%%%

\textsc{MinFAS} was contained in the famous list of 21
NP-complete problems by Karp~\cite{KARP72}. Despite intensive research for
almost four decades, the approximability of this problem remains very
poorly understood due to the big gap between
positive and negative results.
It is known to be at least as hard as vertex cover~\cite{kann92}, but no constant approximation ratio
has been found yet. The best known approximation algorithm achieves
a performance ratio $O(\log n \log \log n)$ \cite{Seymour95,EvenNSS98,EvenNRS00}, where $n$ is the number of vertices of the digraph.
%, and it is based on bounding the integrality gap of the natural covering linear program for \textsc{MinFAS}.
%
Closing this approximability gap is a well-known major open problem in the field of approximation algorithms (see e.g.~\cite{VaziraniBook},
p.~337).
Very recently and conditioned on the Unique Games Conjecture, it was shown~\cite{GuruswamiMR08} that for every constant $C > 0$, it is NP-hard to find a $C$-approximation to the \textsc{MinFAS}.

Important ordering problems can be seen as special cases of
\textsc{MinFAS} with restrictions on the weighting function.
Examples of this kind are provided by ranking problems related to the aggregation of inconsistent information, that have recently received a lot of attention~\cite{Ailon10,abs-1012-3011,AilonCN08,Kenyon-MathieuS07,ZuylenHJW07,ZuylenW09}. Several of these problems can be modeled as (constrained) \textsc{MinFAS} with weights satisfying either \emph{triangle inequalities} (i.e., for any triple $i,j,k$, $w_{(i,j)}+w_{(j,k)}\geq w_{(i,k)}$), or \emph{probability constraints} (i.e., for any pair $i,j$, $w_{(i,j)}+w_{(j,i)}=1$), or both.
Ailon, Charikar and Newman~\cite{AilonCN08} give the first constant-factor randomized approximation algorithm for the unconstrained \textsc{MinFAS} problem with weights that satisfy the triangle inequalities.
%The currently best known constant approximation algorithms for the (constrained) \textsc{MinFAS} with triangle inequalities on the weights can be found in~\cite{Ailon10,ZuylenW09}.
 %
For the same problem Ailon~\cite{Ailon10} gives a $3/2$-approximation algorithm and van Zuylen and Williamson~\cite{ZuylenW09} provide a $2$-approximation algorithm for the constrained version. These are currently the best known results for the (constrained) \textsc{MinFAS} with triangle inequalities and are both based on solving optimally and rounding the linear program relaxation of~\eqref{ILP:FAS}.
  When the probability constraints hold, Mathieu and Schudy~\cite{Kenyon-MathieuS07} obtain a PTAS.
%  With the additional assumption that the input is ``consistent'' with the constraints, i.e., $w_{(i,j)} = 0$ for $(i,j)\in P$, the authors in~\cite{ZuylenW09} also provide a combinatorial 2-approximation algorithms for the constrained \textsc{MinFAS} with triangle inequality.
%

Another prominent special case of
\textsc{MinFAS} with restrictions on the weighting function is given by a classical problem in scheduling, namely {the precedence constrained single machine} scheduling problem to minimize the weighted sum of completion times, denoted as $\prob$ (see e.g.~\cite{LaLeRS93} and~\cite{HSSW97} for a 2-approximation algorithm). This problem can be seen as a constrained \textsc{MinFAS} where the weight of arc $(i,j)$ is equal to the product of two numbers $p_i$ and $ w_j$: $p_i$ is the processing time of job $i$ and $w_j$ is a weight associated to job $j$ (see~\cite{AM09,07AmbMasMutSve,07AmbMasSve,KhotBansalSchedHardness09,CorreaSchulz04,km11} for recent advances). In~\cite{AM09,CorreaSchulz04}, it is shown that the structure of the
weights for this problem allows for all the constraints of size strictly larger than two to
be ignored, therefore the scheduling problem can be seen as a special case of the vertex
cover problem (in normal graphs). The established connection proved later to be very valuable both for
positive and negative results: studying this graph yielded a framework
that unified and improved upon previously best-known
approximation algorithms~\cite{07AmbMasMutSve,06AmbMasSve,km11}; moreover, it helped to obtain the first inapproximability results for this old
problem~\cite{07AmbMasSve,KhotBansalSchedHardness09} by revealing more of its structure and giving a first answer to a long-standing open question~\cite{SchWoe:99}.
%
%In~\cite{AM09,CorreaSchulz04} it is shown that the ILP obtained by removing Constraint~\eqref{vc3P} from~\eqref{ILP:relaxFAS(P)} is a proper formulation for this scheduling problem.The established connection proved very valuable both for positive and negative results~\cite{07AmbMasMutSve,06AmbMasSve,07AmbMasSve}.
%With the additional assumption that the input is ``consistent'' with the constraints, i.e., $w_{(i,j)} = 0$ for $(i,j)\in P$, the authors in~\cite{ZuylenW09} also provide a combinatorial 2-approximation algorithms for the constrained \textsc{MinFAS} with triangle inequality. (Note that Theorem~\ref{th:2vc} removes that assumption.)
%
%%%%%%%%%%%%%%%%%%%%%%%% END LITERATURE %%%%%%%%%%%%%%%%%%%%%%%%
%%%%%%%%%%%%%%%%%%%%%%%% END LITERATURE %%%%%%%%%%%%%%%%%%%%%%%%
\paragraph{New Results.}
The (constrained) \textsc{MinFAS} can be described by the following natural (compact) ILP using linear
ordering variables $\delta_{(i,j)}$ (see e.g.~\cite{ZuylenW09}): variable
$\delta_{(i,j)}$ has value 1 if vertex $i$ precedes vertex $j$ in
the corresponding permutation, and 0 otherwise.
\begin{subequations}
\label{ILP:FAS}
\begin{align}
  \text{[FAS]}\hspace{1cm}  \text{min}& \sum_{i\not =j}\delta_{(i,j)}w_{(i,j)}\\
  \text{s.t.}\hspace{1cm}& \delta_{(i,j)} + \delta_{(j,i)} = 1, &  \text{for all distinct }\ i,j,\label{2loop}\\
  & \delta_{(i,j)}+\delta_{(j,k)}+\delta_{(k,i)}\geq 1, & \text{for all distinct}\ i,j,k, \label{transitivity} \\
    & \delta_{(i,j)}=1,         & \text{for all } (i,j)\in P \label{prec}, \\
  & \delta_{(i,j)}\in \{0,1\}, & \text{for all distinct }\ i, j.
\end{align}
\end{subequations}
Constraint~\eqref{2loop} ensures that in any feasible permutation either vertex $i$
is before $j$ or vice versa. The set of
Constraints~\eqref{transitivity} is used to capture the
transitivity of the ordering relations (i.e., if $i$ is ordered
before $j$ and $j$ before $k$, then $i$ is ordered before $k$, since otherwise by using~\eqref{2loop} we would have $ \delta_{(j,i)}+\delta_{(i,k)}+\delta_{(k,j)}= 0$ violating~\eqref{transitivity}). Constraints~\eqref{prec} ensure that the returned permutation complies with the partial order $P$.
%It is easy to see that [P-IP] is indeed a complete formulation of the problem~\cite{Potts80}.
The constraints in~\eqref{ILP:FAS} were shown to be a minimal equation system 
for the linear ordering polytope in~\cite{grotschel1984acyclic}.

To some extent, one source of difficulty that makes the \textsc{MinFAS} hard to approximate within any constant is provided by the equality in Constraint~\eqref{2loop}. %(this makes the ILP non positive).
To see this, consider, for the time being, the unconstrained \textsc{MinFAS}. The following covering relaxation obtained by relaxing Constraint~\eqref{2loop}  behaves very differently with respect to approximation.
\begin{subequations}
\label{ILP:relaxFAS}
\begin{align}
  \text{min}& \sum_{i\not =j}\delta_{(i,j)}w_{(i,j)}\\
  \text{s.t.}\hspace{1cm}& \delta_{(i,j)} + \delta_{(j,i)} \geq 1, &  \text{for all distinct }\ i,j,\label{vc2loop}\\
  & \delta_{(i,j)}+\delta_{(j,k)}+\delta_{(k,i)}\geq 1, & \text{for all distinct}\ i,j,k, \label{vctransitivity} \\
  & \delta_{(i,j)}\in \{0,1\}, & \text{for all distinct }\ i, j.
\end{align}
\end{subequations}
Problem~\eqref{ILP:relaxFAS} is a special case of  the vertex cover problem in hypergraphs with edges of sizes at most 3. It admits ``easy'' constant approximate solutions (i.e. a trivial primal-dual 3-approximation algorithm, but also a 2-approximation algorithm for \emph{general weights} (no triangle inequalities restrictions) by observing that the associated vertex cover hypergraph is $2$ colorable and using the results in~\cite{AharoniHK96,Krivelevich97}); Vice versa, there are indications that problem~\eqref{ILP:FAS} may not have any constant approximation~\cite{GuruswamiMR08}. Moreover, the fractional relaxation of~\eqref{ILP:relaxFAS}, obtained by dropping the integrality requirement, is a positive linear program and therefore fast NC approximation algorithms exists: Luby and Nisan's algorithm~\cite{LubyN93} computes a feasible $(1+\varepsilon)$-approximate solution in time polynomial in $1/\varepsilon$ and $\log N$, using $O(N)$ processors, where $N$ is the size of the input (fast approximate solution can also be obtained through the methods of~\cite{mor95-plotkin-fast}). On the other side, the linear program relaxation of~\eqref{ILP:FAS} is not positive.
An interesting question is to understand under which assumptions on the weighting function
the covering relaxation~\eqref{ILP:relaxFAS} represents a ``good'' relaxation for \textsc{MinFAS}.

%Surprisingly, despite~\eqref{ILP:relaxFAS} is a ``bad'' relaxation in general,
Surprisingly, we show that the covering relaxation~\eqref{ILP:relaxFAS} is an ``optimal'' relaxation, namely, a \emph{proper} formulation, for the unconstrained \textsc{MinFAS} when the weights satisfy the triangle inequalities. More precisely, we show that any $\alpha$-approximate solution to~\eqref{ILP:relaxFAS} can be turned in polynomial time into an $\alpha$-approximate solution to~\eqref{ILP:FAS}, for any $\alpha\geq 1$ and when the weights satisfy the triangle inequalities. The same claim applies to fractional solutions. (We also observe that the same result does not hold when the weights satisfy the probability constraints (see Appendix~\ref{sec:probability})).

Interestingly, a compact covering formulation can be also obtained for the more general setting with precedence constraints. In this case we need to consider the  following covering relaxation which generalizes~\eqref{ILP:relaxFAS} to partially ordered sets~$\mathbf{P} = (V,P)$.
\begin{subequations}
\label{ILP:relaxFAS(P)}
\begin{align}
 \text{min}& \sum_{i\not =j}\delta_{(i,j)}w_{(i,j)}\\
  \text{s.t.}\hspace{0.2cm} & \delta_{(x_1,y_1)}+\delta_{(x_2,y_2)}\geq 1, &
   (x_2,y_{1}),(x_1,y_{2})\in P,\label{vc2P}\\
  & \delta_{(x_1,y_1)}+\delta_{(x_2,y_2)}+\delta_{(x_3,y_3)}\geq 1, &
  (x_2,y_{1}),(x_3,y_{2}),(x_1,y_{3})\in P,\label{vc3P}\\
  & \delta_{(i,j)}\in \{0,1\}, & (i, j)\in \inc(\mathbf{P}).
\end{align}
\end{subequations}
where $\inc(\mathbf{P}) = \{(x,y)\in V\times V : (x,y),(y,x)\not \in P\}$
is the set of \emph{incomparable pairs} of $\mathbf{P}$.
When the poset is empty, then \eqref{ILP:relaxFAS(P)} boils down to~\eqref{ILP:relaxFAS} (since $P$ is a reflexive binary relation).
Note that \eqref{ILP:relaxFAS(P)} is a relaxation to constrained \textsc{MinFAS}, since Constraint~\eqref{vc2P} and \eqref{vc3P} are valid inequalities (otherwise we would have cycles).
%
%In general ILP~\eqref{ILP:relaxFAS(P)} is a relaxation to constrained \textsc{MinFAS} (if either Constraint~\eqref{vc2P} or \eqref{vc3P} was violated then we would have a cycle). We prove that~\eqref{ILP:relaxFAS(P)} is a proper formulation for constrained \textsc{MinFAS} again when the weights satisfy the triangle inequalities.

Recall that a function $w:V\times V \rightarrow \mathbf{R}$ is \emph{hemimetric} if for all $i,j,k$ the following is satisfied:
\begin{enumerate}
\item $w(i,j)\geq 0$ (\emph{non-negativity}),
\item $w(i,i)= 0$, %(\emph{nullity})
\item $w_{(i,k)}\leq w_{(i,j)}+w_{(j,k)}$ (\emph{triangle inequality}).
\end{enumerate}
The following theorem summarizes the main result of the paper that can be generalized to fractional solutions.
%%%%%%%%%%%%%%%%THEOREM %%%%%%%%%%%%%
\begin{theorem}\label{th:3vc}
If the weighting function $w:V\times V \rightarrow \mathbf{R}$ is hemimetric then any (fractional) solution to~\eqref{ILP:relaxFAS(P)} can be transformed in polynomial time into a feasible (fractional) solution to~\eqref{ILP:FAS} without deteriorating the objective function value.
\end{theorem}
%%%%%%%%%%%%%%%%END THEOREM %%%%%%%%%%%%%%%%%%%%%%%
%%%%%%%%%%%%%%%%%%%%%%%%%%%%%%%%%%%%%%%%%%%%%%

We emphasize that a straightforward application of Theorem~\ref{th:3vc} does not imply a better approximation algorithm for the (constrained) \textsc{MinFAS} with triangle inequality.
However, Theorem~\ref{th:3vc} gives a new surprising structural insight that opens the road to studying the problem under a new light which can benefit from the vast literature and techniques developed for covering problems (this was actually the case for the previously cited scheduling problem~\cite{AM09,07AmbMasMutSve,07AmbMasSve,KhotBansalSchedHardness09,CorreaSchulz04,km11} where the vertex cover insight was essential to obtain improved lower/upper bounds on the approximation ratio). % and that may yield better lower/upper bounds on the approximability of the problem.
Moreover, the fractional relaxation of~\eqref{ILP:relaxFAS(P)}, obtained by dropping the integrality requirement, is a positive linear program and, therefore, we can obtain fast \emph{combinatorial} approximation algorithms that match the best known approximation algorithms up to an arbitrarily small error $\epsilon>0$, by first approximately solving the fractional relaxation of~\eqref{ILP:relaxFAS(P)}~\cite{LubyN93,mor95-plotkin-fast}, then using Theorem~\ref{th:3vc}, and finally applying the rounding algorithms in \cite{Ailon10,ZuylenW09}. (The only known combinatorial approach that matches the best known ratio for the constrained \textsc{MinFAS} with triangle inequality was obtained in~\cite{ZuylenW09} with the additional assumption that the input is ``consistent'' with the constraints, i.e., $w_{(i,j)} = 0$ for $(i,j)\in P$.)

The arguments that we use to prove Theorem~\ref{th:3vc} have some similarities, but also substantial differences from those used to prove the vertex cover nature of problem $\prob$~\cite{AM09}.
The differences come from the diversity of the two weighting functions that make, for example, the scheduling problem without precedence constraints a trivial problem and the (unconstrained) \textsc{MinFAS} with triangle inequality NP-complete.
However, we believe that they both belong to a more general framework, that still has  to be understood, and that may reveal the vertex cover nature of several other natural \textsc{MinFAS} problems (see Section~\ref{Sec:open} for a conjecture).

In the next section we prove Theorem~\ref{th:3vc} by showing how to ``repair'' in polynomial time any feasible solution to~\eqref{ILP:relaxFAS(P)} to obtain a feasible solution to~\eqref{ILP:FAS} that satisfies the claim (similar arguments can be used to generalize the claim to fractional solutions, but details are omitted in this extended abstract).
%
%The proof of Theorem~\ref{th:2vc} is given in Section~\ref{sec:2vc}.
We conclude the paper with a conjecture locating the addressed problem into a general hierarchy within \textsc{MinFAS}.
%Due to space limitations, the omitted proofs can be found in Appendix.
%%%%%%%%%%%%%%%
%\paragraph{Related Works.}
%%%%%%%%%%%%%%%

%%%%%%%%%%%%%%%
%\paragraph{Structure of the paper.}
%%%%%%%%%%%%%%%
%The proof of Theorem~\ref{th:3vc} is given in Section~\ref{sec:3vc}.

%%%%%%%%%%%%%%%%%%%%%%%
\section{Proof of Theorem~\ref{th:3vc}}\label{sec:3vc}
%%%%%%%%%%%%%%%%%%%%%%%

In this section we prove Theorem~\ref{th:3vc} for integral solutions. The proof for fractional solutions is similar and omitted due to space limitations. The structure of the proof is as follows.
Consider any \emph{minimal} integral solution\footnote{Recall that a $0\setminus1$ solution $\delta^*$ is \emph{minimal} if the removal of any arc $(i,j)$ from its support makes it unfeasible.} $\delta^*=\{\delta^*_{(i,j)}:\mbox{for all } i,j\}$ that is feasible to~\eqref{ILP:relaxFAS(P)}, but violates Constraint~\eqref{2loop}. Let us say that pair $\{i,j\}$ is \emph{contradicting} if $\delta_{(i,j)}^*=\delta_{(j,i)}^*=1$.
The violation of Constraint~\eqref{2loop} implies that there exists a non-empty set $A$ of contradicting pairs. The minimality of $\delta^*$ implies that the removal of one of the two arcs of a contradicting pair yields an infeasible solution to~\eqref{ILP:relaxFAS(P)}. The proof works by identifying a subset $A'\subseteq A$ of contradicting pairs, together with another set $B$ of arcs such that, by removing one of the two arcs in any pair from $A'$ and by reverting the arcs in $B$, we obtain a feasible solution to~\eqref{ILP:relaxFAS(P)} with a strictly smaller set of contradicting pairs. Moreover, the new solution is shown to be at least as good as the old one (here we use the assumption that the weighting function is hemimetric). By reiterating the same arguments  we end up with a solution where no contradicting pair exists, i.e. feasible for~\eqref{ILP:FAS}, of value not worse than the initial one.

We start with a preliminary simple observation that characterizes minimal solutions and that will be used several times.% (and whose proof can be found in Appendix~\ref{lemma:prec}).
%
%%%%%%%%%%%%%
\begin{lemma}\label{Lemma:prec}
For any feasible minimal solution $\delta^*=\{\delta^*_{(i,j)}:\mbox{for all } i,j\}$ to~\eqref{ILP:relaxFAS(P)} and any $i,j,k,\ell \in V$ such that $j\not = k$ and $i\not = \ell$, if $\delta^*_{(j,k)}=1$, $\delta^*_{(k,j)}=0$  and $(i,j),(k,\ell)\in P$ then $\delta^*_{(i,\ell)}=1$ and $\delta^*_{(\ell,i)}=0$.
\end{lemma}
%%%%%%%%%%%%%%
\begin{proof}
Note that $\delta_{(i,\ell)}+\delta_{(k,j)}\geq 1$ is part of constraints~\eqref{vc2P}, therefore by the assumptions we have $\delta^*_{(i,\ell)}=1$.

By contradiction, assume that $\delta^*_{(\ell,i)}=1$. By minimality of solution $\delta^*$, there must be a constraint that would be violated if we set $\delta^*_{(\ell,i)}$ to zero.
The latter means that there are incomparable pairs $(x_2,y_2)$ and $(x_3,y_3)$ such that either (i) the following is a valid constraint~\eqref{vc2P} with $\delta^*_{(x_2,y_2)}=0$
$$\delta_{(\ell,i)}+\delta_{(x_2,y_2)}\geq 1,$$
or (ii) the following is a valid constraint~\eqref{vc3P}
$$\delta_{(\ell,i)}+\delta_{(x_2,y_2)}+\delta_{(x_3,y_3)}\geq 1,$$
with $\delta^*_{(x_2,y_2)}=\delta^*_{(x_3,y_3)}=0$.
Case (i) implies that $\delta_{(k,j)}+\delta_{(x_2,y_2)}\geq 1$ is a valid constraint that is violated by solution $\delta^*$.
Similarly, Case (ii) implies that $\delta_{(k,j)}+\delta_{(x_2,y_2)}+\delta_{(x_3,y_3)}\geq 1$ is a valid constraint that is violated by solution $\delta^*$.
\qd
\end{proof}
%%%%%%%%%%%%%%%%%%%%%%%% END PROOF %%%%%%%%%%%%%%%%%%%%%%%%%%%

Let $\delta^*=\{\delta^*_{(i,j)}:\mbox{for all } i,j\}$ be an $\alpha$-approximate minimal solution to~\eqref{ILP:relaxFAS(P)}.
For any triple $(a,c,b)\in V^3$ of distinct vertices, we say that $(a,c,b)$ is a \emph{basic triple} if the following holds (see Fig.~\ref{fig:basic_triple}): $\delta^*_{(a,c)}= \delta^*_{(c,b)}=\delta^*_{(a,b)}= \delta^*_{(b,a)}=1$ and $\delta^*_{(c,a)}= \delta^*_{(b,c)}=0$. Let $T$ be the set of all the basic triples. The following lemma states that basic triples are witnesses of infeasibility.
 \begin{figure}[hbtp]
    \centering
    \includegraphics[width=3cm]{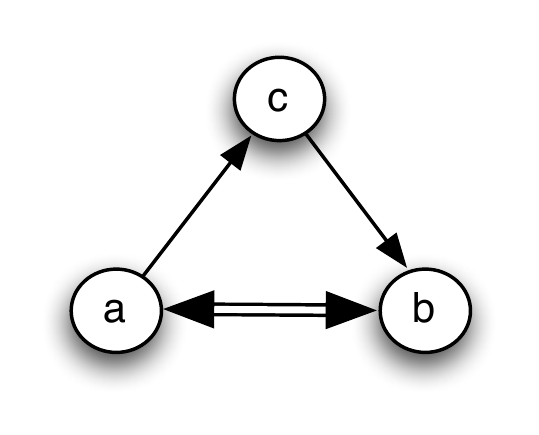}
    \caption{Basic triple: $\delta^*_{(a,c)}= \delta^*_{(c,b)}=\delta^*_{(a,b)}= \delta^*_{(b,a)}=1$, $\delta^*_{(c,a)}= \delta^*_{(b,c)}=0$. }
  \label{fig:basic_triple}
\end{figure}
%%%%%%%%%
%%%%%%%
\begin{lemma}\label{lemma:T}
If solution $\delta^*$ is a minimal solution to~\eqref{ILP:relaxFAS(P)} but not feasible to~\eqref{ILP:FAS}, then $T\not=\emptyset$.
\end{lemma}
%%%%%%%%%
\begin{proof}
Assume that $\delta^*_{(a,b)}= \delta^*_{(b,a)}=1$. Variable $\delta^*_{(a,b)}$ cannot be turned to zero because there exists $c,d,e,f\in V$ such that $\delta^*_{(c,d)}= \delta^*_{(e,f)}=0$ and the following is a valid constraint~\eqref{vc3P}
$$\delta_{(a,b)}+\delta_{(c,d)}+\delta_{(e,f)}\geq 1.$$
By a simple application of Lemma~\ref{Lemma:prec} (see Fig.~\ref{fig:triple_existence}) it follows that $(a,b,d)$ is a basic triple.
 \begin{figure}[hbtp]
    \centering
    \includegraphics[width=3cm]{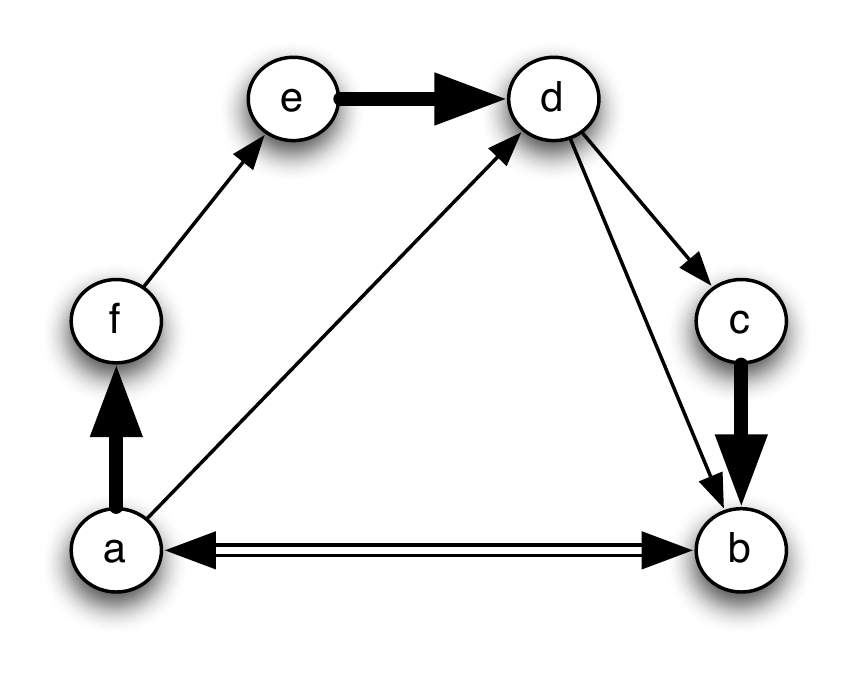}
    \caption{Existence of a basic triple $(a,d,b)$ assuming $\delta^*_{(a,b)}= \delta^*_{(b,a)}=1$. Bold arrows represent poset relationship, namely $(a,f),(e,d),(c,b)\in P$.}
  \label{fig:triple_existence}
\end{figure}
\qd
\end{proof}
%%%%%%%

For any given vertex $v$, let us define the following set of arcs that will be used to ``drop and reverse'' arcs in a synchronized way to obtain new solutions:
\begin{eqnarray}
S_v&=&\{(i,j): (v,i,j)\in T\} \label{S_v}.\\
M_v&=&\{(i,j): (j,v,i)\in T\} \label{M_v}.\\
E_v&=&\{(i,j): (i,j,v)\in T\} \label{E_v}.
\end{eqnarray}
Note that $S_v$, $M_v$ and $E_v$ are pairwise disjoint.
%The following properties are easy to check.
%\begin{property}\label{mse_prop}
%For any $v\in V$,
%\begin{enumerate}
%\item If $(x,y)\in S_v \cup M_v$ then $\delta^*_{(x,v)}=\delta^*_{(v,y)}=0$.\label{pMS}
%\item If $(x,y)\in M_v$ then $\delta^*_{(x,v)}=0$.\label{pSv}
%\item $S_v\cap E_v=\emptyset$.\label{SvcupEv}
%\end{enumerate}
%\end{property}

%%%%%%%%%%%%%%% LEMMA %%%%%%%%%%%%%%%%%%%%%%%%%%
\begin{lemma}\label{th:newsol}
For any $v\in V$ and $X\in\{S_v,E_v\}$, solution $\delta^X=\{\delta^X_{(i,j)}:\mbox{for all } i,j\}$ as defined in the following is a feasible solution for~\eqref{ILP:relaxFAS(P)}:
\begin{enumerate}
\item $\delta^X_{(i,j)}=0$ for each $(i,j)\in M_v$.
\item $\delta^X_{(i,j)}=0$ and $\delta^X_{(j,i)}=1$ for each $(i,j)\in X$.
\item $\delta^X_{(i,j)}=\delta^*_{(i,j)}$ elsewhere.
\end{enumerate}
\end{lemma}
%%%%%%%%%%%%%% END LEMMA %%%%%%%%%%%%%%%%%%%%%%%
\begin{proof}
We start showing that solution $\delta^X$ satisfies the second set~\eqref{vc3P} of constraints in~\eqref{ILP:relaxFAS(P)} for any $X\in\{S_v,E_v\}$. The proof that $\delta^X$ satisfies the first set of constraints~\eqref{vc2P} is similar.

Let us assume that $X=S_v$ (the proof for $X=E_v$ is symmetric).
 %and it can be found in Appendix~\ref{omitted} for an easy checking
  Since solution $\delta^X$ is obtained from the feasible solution $\delta^*$ by switching some variables to zero and others to one, we might violate only those constraints with at least one variable from $\delta^X$ that is turned to zero, i.e. the set of constraints that have at least one variable from $\{\delta^X_{(i,j)} :(i,j)\in X\cup M_v\}$. Let $(i,j')\in X\cup M_v$ and for any $j,k',k,i'\in V$ such that $\delta_{(i,j')}+\delta_{(j,k')} +\delta_{(k,i')} \geq 1$ is a valid constraint~\eqref{vc3P}, we want to prove that the following holds:
\begin{equation}
\delta^X_{(i,j')}+\delta^X_{(j,k')} +\delta^X_{(k,i')} \geq 1.
\end{equation}

We distinguish between the following cases:
\begin{itemize}
\item[Case (a):] $\delta^*_{(j,k')}=1$. Since $(i,j')\in S_v\cup M_v$ then $\delta^*_{(j',v)}=1$ (see Fig.~\ref{fig:caseA}).

 If $(i,j')\in M_v$ then $\delta^*_{(j',v)}=1$ and $\delta^*_{(v,j')}=0$. By applying Lemma~\ref{Lemma:prec} we can conclude that  $\delta^*_{(j,v)}=1$.

If $(i,j')\in S_v$ we claim that  $\delta^*_{(j,v)}=1$ as well. By contradiction assume  $\delta^*_{(j,v)}=0$ and therefore  $\delta^*_{(v,j)}=1$. By applying Lemma~\ref{Lemma:prec} we would have $\delta^*_{(v,j')}=1$ and $\delta^*_{(j',v)}=0$. The latter contradicts the assumption that $(i,j')\in S_v$.

Since $\delta^*_{(j,v)}=1$, we have $(j,k')\not \in S_v\cup M_v$ (since if $(j,k')\in S_v\cup M_v$ then $\delta^*_{(j,v)}=0$ as shown in Fig.~\ref{fig:caseA})   and therefore $\delta^X_{(j,k')}=\delta^*_{(j,k')}=1$.
%%%%%%%
\begin{comment}
\begin{figure}[hbtp]
    \centering
    \includegraphics[width=6cm]{caseAP}
    \caption{Case (a).}
  \label{fig:caseA}
\end{figure}
\end{comment}
%%%%%%%%%
\item[Case (b):] $\delta^*_{(k,i')}=1$. Since $(i,j')\in S_v\cup M_v$ then $\delta^*_{(i,v)}=0$ (see Fig.~\ref{fig:caseB}) and $\delta^*_{(i',v)}=0$ by Lemma~\ref{Lemma:prec}. The latter implies that $(k,i')\not \in S_v\cup M_v$ (since if $(k,i')\in S_v\cup M_v$ then $\delta^*_{(i',v)}=1$ as shown in Fig.~\ref{fig:caseB}) and therefore $\delta^X_{(k,i')}=\delta^*_{(k,i')}=1$.
%
\begin{comment}
\begin{figure}[hbtp]
    \centering
    \includegraphics[width=6cm]{caseBP}
    \caption{Case (b).}
  \label{fig:caseB}
\end{figure}
\end{comment}
%
%\pagebreak
\item[Case (c):] $\delta^*_{(j,k')}=\delta^*_{(k,i')}=0$. Under the current assumption, by Lemma~\ref{Lemma:prec} and constraint~\eqref{vc3P}, it is easy to check that $\delta^*_{(v,k)}=1$ (see  Fig.~\ref{fig:caseC}). We distinguish between two subcases: (i) $\delta^*_{(k,v)}=1$ (if $(i,j')\in M_v$ this is the only possible case) and (ii) $\delta^*_{(k,v)}=0$. If (i) holds then $(i',k)\in S_v$ and therefore $\delta^X_{(k,i')}=1$. Otherwise, by Lemma~\ref{Lemma:prec} we have $\delta^*_{(v,k')}=1$ and $\delta^*_{(k',v)}=0$. Moreover, since under (ii) we have $\delta^*_{(v,j')}=\delta^*_{(j',v)}=1$, by minimality of the solution, there exists a node $q$ such that $\delta^*_{(j',q)}=\delta^*_{(q,v)}=1$ and $\delta^*_{(q,j')}=\delta^*_{(v,q)}=0$. By applying Lemma~\ref{Lemma:prec} we have $\delta^*_{(j,q)}=1$ and $\delta^*_{(q,j)}=0$. Therefore, $\delta^*_{(v,k')}=\delta^*_{(k',j)}=\delta^*_{(j,q)}=\delta^*_{(q,v)}=1$ and $\delta^*_{(k',v)}=\delta^*_{(v,q)}=\delta^*_{(q,j)}=\delta^*_{(j,k')}=0$ imply that $(k',j)\in S_v$ which implies that $\delta^X_{(j,k')}=1$.
%
\begin{comment}
\begin{figure}[hbtp]
    \centering
    \includegraphics[width=6cm]{caseCP}
    \caption{Case (c).}
  \label{fig:caseC}
\end{figure}
\end{comment}
\end{itemize}
%
%Assume that we want to remove all the arcs in $M_v$. It is not difficult to check (see also Fig.~\ref{fig:v}.A) that for each arc $(b,a)\in M_v$ there is a blocking vertex $c$ with $(b,c)\in S_v$ (or $(c,a)\in E_v$). Therefore if we want to remove arc $(b,a)\in M_v$ we have to remove the ``blocking situation'' by reverting $(b,c)$ (or $(c,a)$). We show in the following that for \emph{any} arc $(b,c)\in S_v$ the blocking arcs are also inside $S_v$ (similar argument for \emph{any} arc $(c,a)\in E_v$). Consider \emph{any} arc $(b,c)\in S_v$ and assume that we want to revert it. The latter might be prevented by the presence of a blocking vertex $x$ for $(b,c)$ (see Fig.~\ref{fig:v}.B). Since $v$ and $x$ are blocking vertices for $(a,b)$ and $(b,c)$, respectively, then $\delta^*_{(v,b)}=\delta^*_{(b,x)}=1$ and $\delta^*_{(b,v)}=\delta^*_{(x,b)}=0$, and by the triangle constraint we must have $\delta^*_{(v,x)}=1$. Now, we have two possibilities: either (i) $\delta^*_{(x,v)}=1$ or (ii) $\delta^*_{(x,v)}=0$. If (i) holds, then $(v,b,x)\in T$ and therefore $(b,x)\in S_v$, otherwise (ii) implies $(v,x,c)\in T$ and therefore $(x,c)\in S_v$. Basically, what we have shown is that for any arc $(b,c)\in S_v$ and for any blocking vertex $x$ for arc $(b,c)$, either $(b,x)\in S_v$ or $(x,c)\in S_v$, and therefore we create a feasible solution by reverting all the arcs in $S_v$. The following lemma follows.
\qd
\end{proof}

According to solution $\delta^*$, let us say that pair $\{i,j\}$ is \emph{contradicting} if $\delta_{(i,j)}^*=\delta_{(j,i)}^*=1$.
By Lemma~\ref{th:newsol}, any solution $\delta'\in\Delta=\{\delta^X:v\in V \mbox{ and } X\in \{S_v,E_v\}\}$ is a feasible solution for~(\ref{ILP:relaxFAS(P)}). Moreover, it is easy to observe that $\delta'$ has a strictly smaller number of contradicting pairs.

The claim of Theorem~\ref{th:3vc} follows by proving the following Lemma~\ref{th:better} which shows that among the feasible solutions in $\Delta$ there exists one whose value is not worse than the value of $\delta^*$. Therefore, after at most $O(|V|^2)$ ``applications'' of Lemma~\ref{th:better} we end up with a solution where no contradicting pair exists, i.e. feasible for~\eqref{ILP:FAS}.

%%%%%%%%%%%%%%%%%%%%%%%%%%%%%% LEMMA %%%%%%%%%%%%%%%%%%%%
\begin{lemma}\label{th:better}
If $\delta^*$ is not a feasible solution for~\eqref{ILP:FAS} then there exists a feasible solution for~\eqref{ILP:relaxFAS(P)} in $\Delta=\{\delta^X:v\in V \mbox{ and } X\in \{S_v,E_v\}\}$ whose value is not worse than the value of $\delta^*$.
\end{lemma}
%%%%%%%%%%%%%%%%%%%%%%%%%%%%% END LEMMA %%%%%%%%%%%%%%%%%
\begin{proof}
By contradiction, we assume that every solution in $\Delta$ has value worse than $\delta^*$.

By Lemma~\ref{th:newsol}, for any vertex $v$ we can obtain two feasible solutions by removing all the arcs from $M_v$ and reverting, alternatively, either all the arcs from $S_v$, or all the arcs from $E_v$. Since we are assuming that every solution in $\Delta$ has value worse than $\delta^*$, the following two inequalities express the latter for any $v\in V$.
\begin{eqnarray}
\sum_{(b,a)\in M_v}w_{(b,a)}+\sum_{(i,j)\in S_v}w_{(i,j)} &<& \sum_{(i,j)\in S_v}w_{(j,i)}, \label{sv}\\
\sum_{(b,a)\in M_v}w_{(b,a)}+\sum_{(i,j)\in E_v}w_{(i,j)} &<& \sum_{(i,j)\in E_v}w_{(j,i)}. \label{ev}
\end{eqnarray}
By summing~(\ref{sv}) and~(\ref{ev}) for all $v$ we obtain the following valid inequality:
\begin{equation}\label{LR1}
\underbrace{\sum_{v\in V} \left(2\cdot \sum_{(b,a)\in M_v}w_{(b,a)}+\sum_{(i,j)\in S_v\cup E_v}w_{(i,j)}\right)}_{LHS(1)} < \underbrace{\sum_{v\in V} \left(\sum_{(i,j)\in S_v \cup E_v}w_{(j,i)}\right).}_{RHS(1)}
\end{equation}

%%%%%%%%%%%%%%%%%%%%%%%%%%%%%
\paragraph{A Triangle Inequality Condition.}
For any basic triple $(a,c,b)\in T$ we consider the following two valid triangle inequalities.
\begin{eqnarray}
w_{(c,a)}&\leq& w_{(c,b)}+w_{(b,a)},\label{t1}\\
w_{(b,c)}&\leq& w_{(b,a)}+w_{(a,c)}.\label{t2}
\end{eqnarray}

By summing~(\ref{t1}) and~(\ref{t2}) for all $(a,c,b)\in T$ we obtain the following valid inequality:

\begin{equation}\label{LR2}
\underbrace{\sum_{(a,c,b)\in T} \left(w_{(b,c)}+w_{(c,a)}\right)}_{LHS(2)}
\leq
\underbrace{\sum_{(a,c,b)\in T} \left(2\cdot w_{(b,a)}+w_{(a,c)}+w_{(c,b)}\right).}_{RHS(2)}
\end{equation}

\paragraph{The Contradiction.}
Note that for every $(a,c,b)\in T$ we have $(a,c)\in E_b$ and $(c,b)\in S_a$. Therefore:

\begin{eqnarray}
LHS(2)&=&\sum_{(a,c,b)\in T} \left(w_{(b,c)}+w_{(c,a)}\right)\nonumber \\
&=&\sum_{v\in V} \left( \sum_{(i,j):(v,i,j)\in T} w_{(j,i)} + \sum_{(i,j):(i,j,v)\in T}  w_{(j,i)}  \right)\nonumber \\
&\stackrel{(\ref{S_v}),(\ref{E_v})}{=}& \sum_{v\in V} \left(\sum_{(i,j)\in S_v \cup E_v}w_{(j,i)}\right)=RHS(1). \label{R1=L2}
\end{eqnarray}
Therefore, by~(\ref{LR1}), (\ref{LR2}) and (\ref{R1=L2}) we have
$LHS(1)< RHS(1)=LHS(2)\leq RHS(2)$.
We get a contradiction by showing that $RHS(2)=LHS(1)$:
\begin{eqnarray}
RHS(2)&=& \sum_{(a,c,b)\in T} \left(2\cdot w_{(b,a)}+w_{(a,c)}+w_{(c,b)}\right) \nonumber \\
&=&\sum_{v\in V} \left( 2\cdot \sum_{(a,b):(a,v,b)\in T} w_{(b,a)}+ \sum_{(i,j):(v,i,j)\in T} w_{(i,j)} + \sum_{(i,j):(i,j,v)\in T}  w_{(i,j)}  \right)\nonumber \\
&\stackrel{(\ref{S_v}),(\ref{E_v}),(\ref{M_v})}{=}& \sum_{v\in V} \left(2\cdot \sum_{(b,a)\in M_v}w_{(b,a)}+\sum_{(i,j)\in S_v\cup E_v}w_{(i,j)}\right)=LHS(1). \label{R2=L1}
\end{eqnarray}
\qd
\end{proof}
%%%%%%%%%%%% END PROOF %%%%%%%%%%%%%%%%%%%

%%%%%%%%%%%%%%%%%%%%%%%%%%%%%%%%%%%%%%%%%%%%%
%%%%%%%%%%%%%%%%%%%%%%%%%%%%%%%%%%%%%%%%%%%%%
\section{Future directions}\label{Sec:open}
%%%%%%%%%%%%%%%%%%%%%%%%%%%%%%%%%%%%%%%%%%%%%
%%%%%%%%%%%%%%%%%%%%%%%%%%%%%%%%%%%%%%%%%%%%%
The constrained \textsc{MinFAS} problem admits a natural covering formulation with an exponential number of constraints (see
e.g.~\cite{amms-pcscdtpo-08}):

\begin{subequations}
\label{ILP:HPVC}
\begin{align}
  %\text{\VC{}}\hspace{1cm}
  \text{min}&
  \sum_{(i,j)}\delta_{(i,j)}w_{(i,j)}\\
  \text{s.t.}& \sum_{i=1}^c\delta_{(x_i,y_i)}\geq 1, &
  \text{for all}\ c\geq 2,\ (x_i,y_i)_{i=1}^c\ \text{s.t.}\  (x_i,y_{i+1})\in P,\label{con:altern_cycle}\\
  & \delta_{(i,j)}\in \{0,1\}, & (i, j)\in \inc(\mathbf{P}).
\end{align}
\end{subequations}
The condition $(x_i,y_{i+1})\in P$ in
constraint~(\ref{con:altern_cycle}) is to be read cyclically,
namely, $(x_c,y_1)\in P$. The
hyperedges in this vertex cover problem are exactly the alternating cycles of poset
$P$ (see e.g.~\cite{t-cpos-92}).

In this paper we prove that when the weights satisfy the triangle inequality then we can drop from~\eqref{ILP:HPVC} all the constraints of size strictly larger than three.
%We say that the weights satisfy the $k$-gonal inequalities if the weighting function is $(1,k)$-near-metric.
Generalizing, it would be nice to prove/disprove the following statement that we conjecture to be true.

\begin{hypothesis} When the weights satisfy the $k$-gonal inequalities, i.e., if for all $a_1,\ldots,a_{k}\in V$ we have
$w_{(a_1,a_k)}\leq w_{(a_1,a_2)}+\ldots +w_{(a_{k-1},a_k)}$, then there exists a constant $c(k)$, whose value depends on $k$, such that a proper formulation for the constrained \textsc{MinFAS} problem can be obtained by dropping from~\eqref{ILP:HPVC} all the constraints of size strictly larger than $c(k)$.
\end{hypothesis}
\textsc{MinFAS} problems with weights belonging to interval $[1,k-1]$ are examples of problems with $k$-gonal inequalities on the weights.
If true, the above structural result has the important implication that, for any constant $k$, constrained \textsc{MinFAS} with $k$-gonal inequalities on the weights admits a constant approximation algorithm (in contrast to the general case with arbitrary $k$ that does not seem to have any constant approximation assuming the Unique Games Conjecture~\cite{GuruswamiMR08}).

%Moreover, it would be nice to use the large literature and techniques developed for covering problems to improve the best known ratios for \textsc{MinFAS} with triangle inequalities on the weights. This was actually the case for the scheduling problem $\prob$: in~\cite{AM09,CorreaSchulz04} it was first shown that the structure of the weights for this problem allows for all the constraints of size strictly larger than two to be ignored, therefore the scheduling problem can be seen as a special case of the vertex cover problem. The established connection proved later to be very valuable both for positive and negative results: studying this graph yielded a framework that unified and improved upon previously best-known approximation algorithms~\cite{07AmbMasMutSve,06AmbMasSve}; moreover, it helped to obtain the first inapproximability results for this old problem~\cite{07AmbMasSve,KhotBansalSchedHardness09} by revealing more of its structure.

%\begin{comment}
\paragraph{Acknowledgments.}
I'm indebted with Nikos Mutsanas for many useful discussions and comments.
This research is supported by
Swiss National Science Foundation project N. 200020-122110/1 ``Approximation Algorithms for Machine Scheduling Through Theory and Experiments III'' and by Hasler Foundation Grant 11099.
%\end{comment}
{\small
\bibliographystyle{abbrv}
\bibliography{ref,MonaldoPublications,sdp-ordering}
}

\clearpage

%%%%%%%%%%%%%%%%%%%%%%%%%%%%%%%%%%%%%%%%%%%%%
%%%%%%%%%%%%%%%%%%%%%%%%%%%%%%%%%%%%%%%%%%%%%
%%%%%%%%%%%%%%%%%%%%%%%%%%%%%%%%%%%%%%%%%%%%%
\begin{appendix}
%%%%%%%%%%%%%%%%%%%%%%%%%%%%%%%%%%%%%%%%%%%%%
%%%%%%%%%%%%%%%%%%%%%%%%%%%%%%%%%%%%%%%%%%%%%
%%%%%%%%%%%%%%%%%%%%%%%%%%%%%%%%%%%%%%%%%%%%%
\noindent{\Large \bf Appendix}

%%%%%%%%%%%%%%%%%%%%%%%%%%%%%%%%%%%%%%%%%%%%%%%%
%%%%%%%%%%%%%%%%%%%%%%%%%%%%%%%%%%%%%%%%%%%%%%
\section{Ranking with probability inequalities: a counterexample}\label{sec:probability}
%%%%%%%%%%%%%%%%%%%%%%%%%%%%%%%%%%%%%%%%%%%%%%%%
%%%%%%%%%%%%%%%%%%%%%%%%%%%%%%%%%%%%%%%%%%%%%%
The following example shows that probabilities inequalities are not sufficient for~\eqref{ILP:relaxFAS(P)} to be a proper formulation:

$$w_{(i,j)}+w_{(j,i)}=1 \mbox{ for all distinct } i,j$$

Consider the instance with 8 nodes with weight zero on the arcs displayed in Fig.~\ref{fig:prob} (therefore the reversed arcs have weight $1$). Moreover, all the arcs in $\{2,3\}\times \{ 7,8 \}$ have weight 1 (the reversed zero). Finally, all the remaining arcs have weight $0.5$, namely those in $\{1\}\times \{ 4,5,6 \}$ and the reversed ones.
\begin{figure}[hbtp]
    \centering
    \includegraphics[width=4cm]{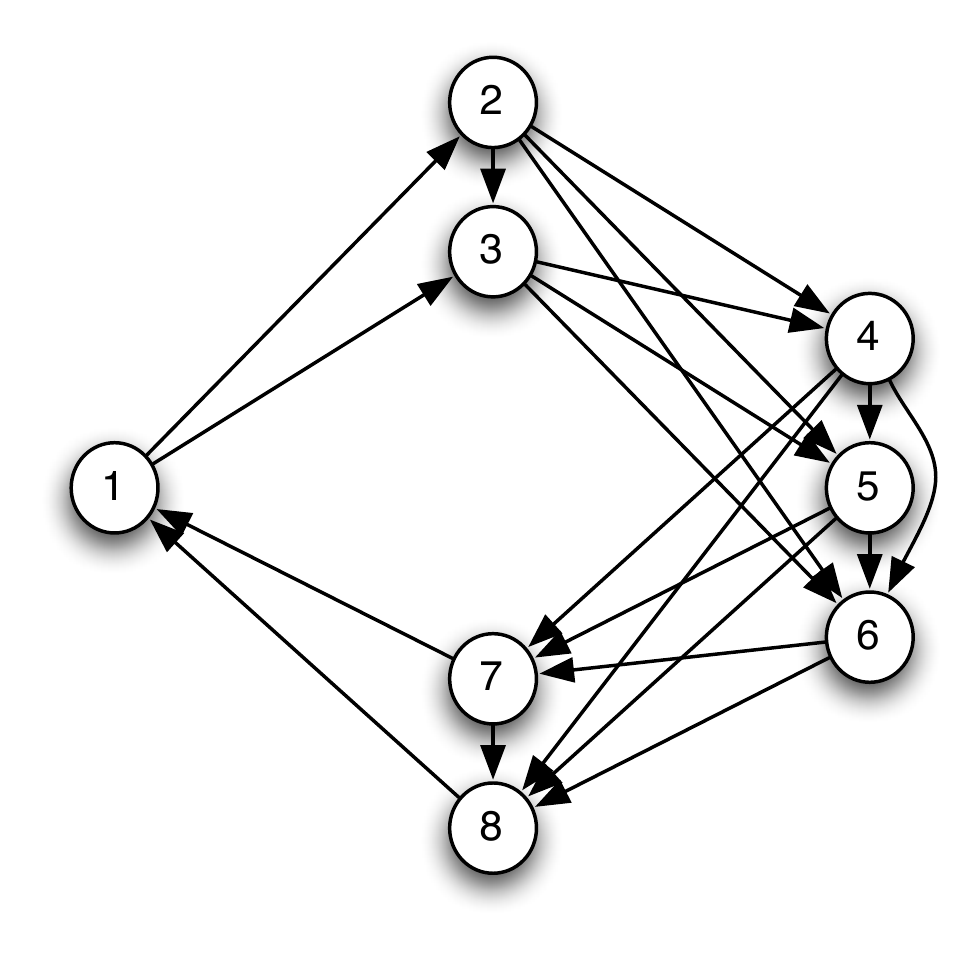}
    \caption{Counterexample for probability inequalities. }
  \label{fig:prob}
\end{figure}
A feasible solution for~\eqref{ILP:relaxFAS} is obtained by picking all the displayed arcs in Fig.~\ref{fig:prob} and none of the reversed ones (therefore we have to pick also those in $\{2,3\}\times \{ 7,8 \}$, $\{ 7,8 \} \times \{2,3\}$,  $\{ 4,5,6 \}\times \{1\}$ and  $\{1\}\times \{ 4,5,6 \}$ in order to satisfy the constraints in~\eqref{ILP:relaxFAS}). This solution has value $7$, whereas any total ordering has value not smaller than $7.5$ (the best total ordering is $(2,3,4,5,6,7,8,1)$).

%%%%%%%%%%%%%%%%%%%%%%%%%%%%%%%%%%%%%%%%%%%%%%%%
%%%%%%%%%%%%%%%%%%%%%%%%%%%%%%%%%%%%%%%%%%%%%%
%\section{Proof of Lemma~\ref{Lemma:prec}}\label{lemma:prec}
%%%%%%%%%%%%%%%%%%%%%%%%%%%%%%%%%%%%%%%%%%%%%%%%
%%%%%%%%%%%%%%%%%%%%%%%%%%%%%%%%%%%%%%%%%%%%%%

%%%%%%%%%%%%%%%%%%%%%%%%%%%%%%%%%%%%%%%%%%%%%%%%
%%%%%%%%%%%%%%%%%%%%%%%%%%%%%%%%%%%%%%%%%%%%%%
\section{A comment on formulation~\eqref{ILP:relaxFAS(P)}}
%%%%%%%%%%%%%%%%%%%%%%%%%%%%%%%%%%%%%%%%%%%%%%%%
%%%%%%%%%%%%%%%%%%%%%%%%%%%%%%%%%%%%%%%%%%%%%%

If the poset is not empty the additional constraints that are present in formulation~\eqref{ILP:relaxFAS(P)} but not in~\eqref{ILP:relaxFAS} are also necessary. Indeed, in Figure~\ref{fig:four_prec} any permutation that complies with the precedence constraints has value larger than the solution suggested in the picture with a cycle.

 \begin{figure}[hbtp]
    \centering
    \includegraphics[width=10cm]{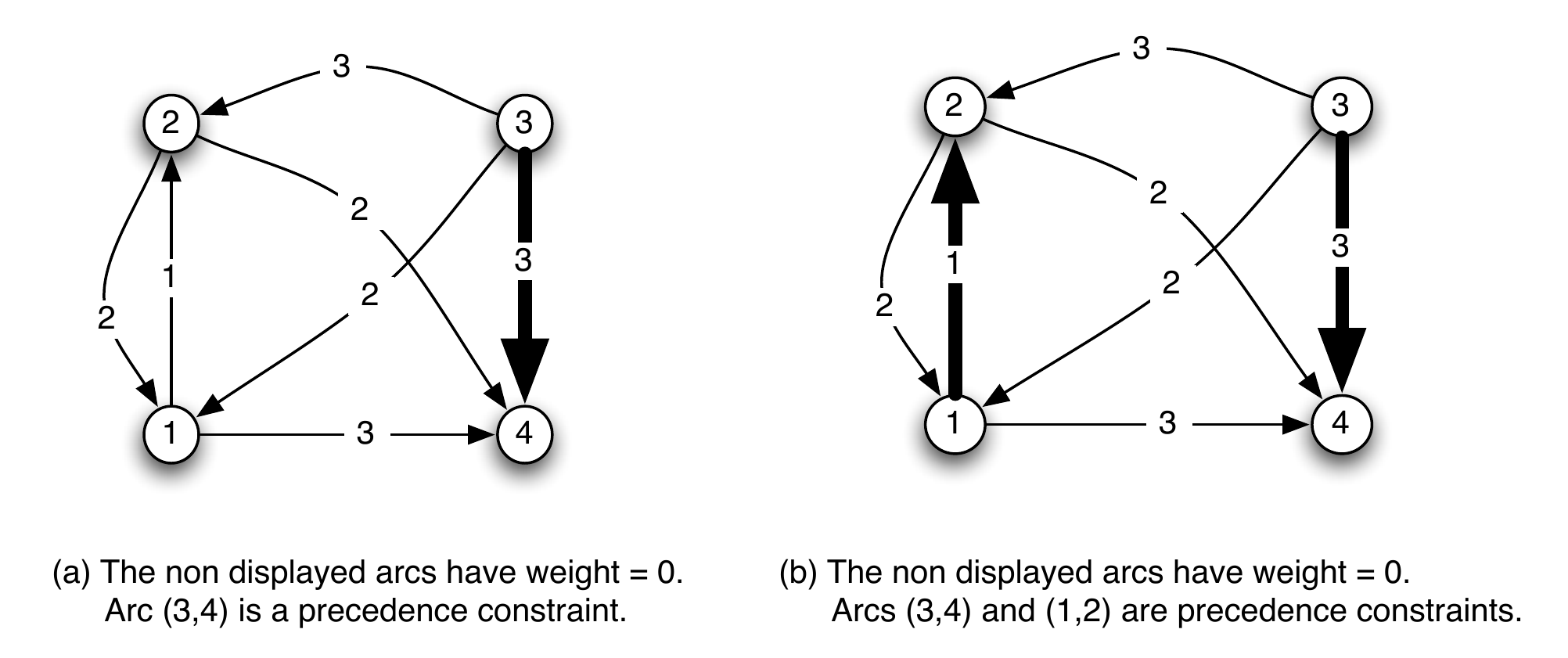}
    \caption{Solution $\delta^*_{(1,2)}= \delta^*_{(2,3)}=\delta^*_{(3,4)}= \delta^*_{(4,1)}=\delta^*_{(1,3)}= \delta^*_{(3,1)}=\delta^*_{(2,4)}= \delta^*_{(4,2)}=1$ has value smaller than any valid permutation.}
  \label{fig:four_prec}
\end{figure}

\pagebreak
%%%%%%%%%%%%%%%%%%%%%%%%%%%%%%%%%%%%%%%%%%%%%%%%
%%%%%%%%%%%%%%%%%%%%%%%%%%%%%%%%%%%%%%%%%%%%%%%%
%%%%%%%%%%%%%%%%%%%%%%%%%%%%%%%%%%%%%%%%%%%%%%%%
\section{Figures used in the proof of Lemma~\ref{th:newsol}}
%%%%%%%%%%%%%%%%%%%%%%%%%%%%%%%%%%%%%%%%%%%%%%%%
%%%%%%%%%%%%%%%%%%%%%%%%%%%%%%%%%%%%%%%%%%%%%%%%
%%%%%%%%%%%%%%%%%%%%%%%%%%%%%%%%%%%%%%%%%%%%%%%%

\begin{figure}[hbtp]
    \centering
    \includegraphics[width=9cm]{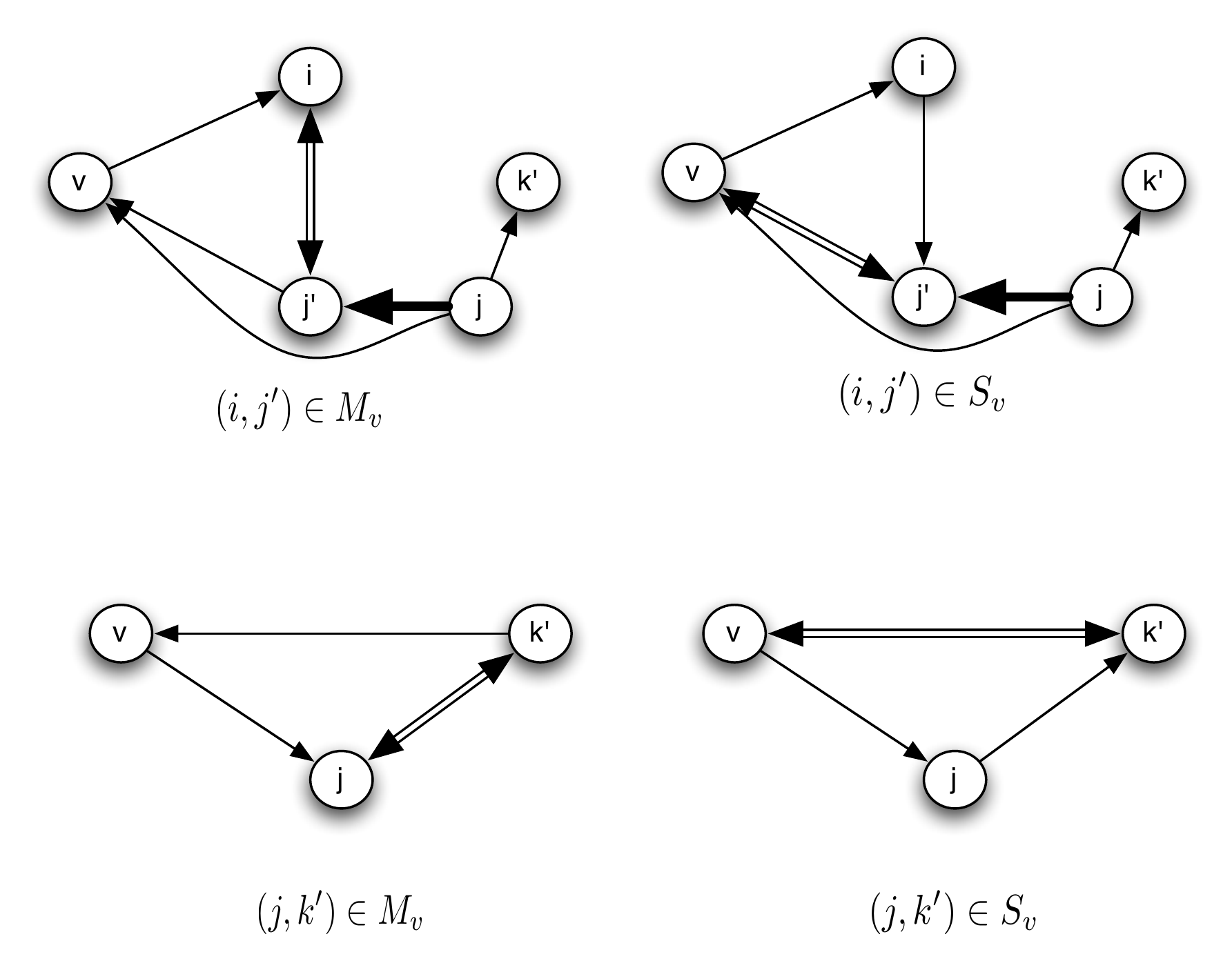}
    \caption{Case (a).}
  \label{fig:caseA}
\end{figure}

\begin{figure}[hbtp]
    \centering
    \includegraphics[width=9cm]{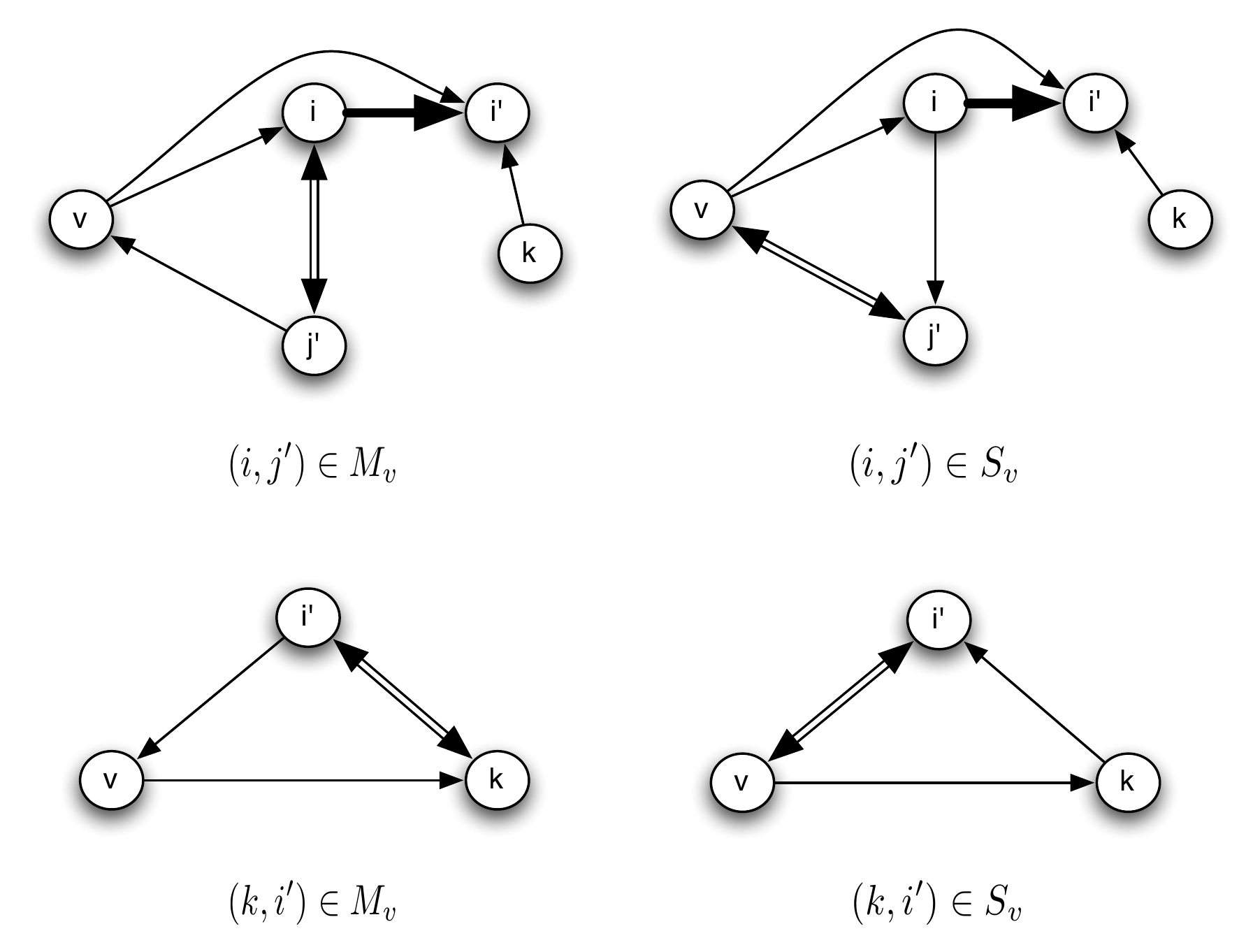}
    \caption{Case (b).}
  \label{fig:caseB}
\end{figure}
\begin{figure}[hbtp]
    \centering
    \includegraphics[width=9cm]{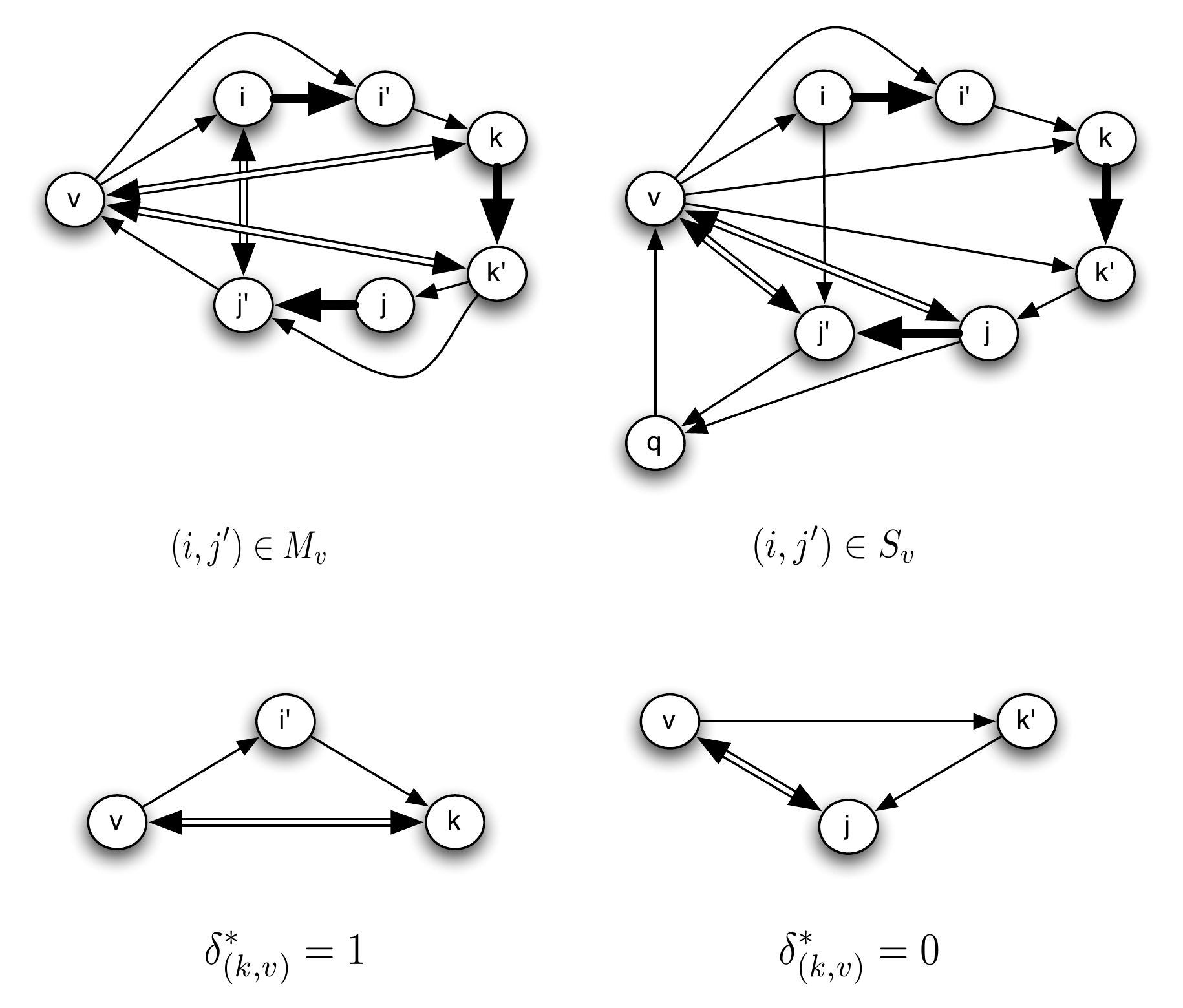}
    \caption{Case (c).}
  \label{fig:caseC}
\end{figure}

\end{appendix}

\end{document}